\documentclass[12pt]{article}
\usepackage{arxiv}
\usepackage[colorlinks=true,citecolor=blue]{hyperref}
\usepackage{graphicx}
\usepackage{amsmath,amssymb}
\usepackage{amsthm}
\usepackage{bm}
\usepackage{stmaryrd}
\usepackage{dsfont}
\newtheorem{theorem}{\bf Theorem}[section]
\newtheorem{definition}{\bf Definition}[section]
\newtheorem{proposition}{\bf Proposition}[section]

\def\Wr{\mathop{\mathrm{Wr}}\nolimits}

\DeclareMathOperator{\sech}{sech}

\begin{document}

\title{Modulation theory for soliton resonance and Mach reflection}

\author{Samuel J. Ryskamp\thanks{ Department of Applied Mathematics, University 
of Colorado, Boulder, CO (samuel.ryskamp@colorado.edu)} \and Mark A. Hoefer\thanks{ Department of Applied Mathematics, University 
of Colorado, Boulder, CO} \and Gino Biondini\thanks{Department of Mathematics and Department of Physics, State University of New York,  Buffalo, NY}}

\maketitle

\begin{abstract}
Resonant Y-shaped soliton solutions to the Kadomtsev-Petviashvili II (KPII) equation are modelled as shock solutions to an infinite family of modulation conservation laws. The fully two-dimensional soliton modulation equations, valid in the zero dispersion limit of the KPII equation, are demonstrated to reduce to a one-dimensional system. In this same limit, the rapid transition from the larger Y soliton stem to the two smaller legs limits to a travelling discontinuity. This discontinuity is a multivalued, weak solution satisfying modified Rankine-Hugoniot jump conditions for the one-dimensional modulation equations. These results are applied to analytically describe the dynamics of the Mach reflection problem, V-shaped initial conditions that correspond to a soliton incident upon an inward oblique corner. Modulation theory results show excellent agreement with direct KPII numerical simulation.
\end{abstract}

\section{Introduction}
Mach reflection occurs when a sufficiently large amplitude line soliton or classical shock interacts with a barrier at a sufficiently small angle. A Y-shaped triad is formed consisting of two smaller amplitude solitons or shocks and a larger "Mach" stem  perpendicular to the barrier (see figure~\ref{fig:mach_ref_schematic}).  This phenomenon was first reported experimentally in the context of shallow water solitons impinging on a corner in J. Scott Russell's seminal 1845 paper 1845 \cite{russell_1845}, while later in 1875 Ernst Mach observed his eponymous phenomenon arising through interacting shocks in gas dynamics  \cite{mach_uber_1875,krehl_1991}. In this paper, we present a method for modelling Y-shaped and X-shaped solitons as shock solutions to soliton modulation equations. 

Line solitons are special traveling wave solutions of the Kadomstev-Petviashvili (KP) equation 
\begin{equation}
    \label{eq:kp}
    (u_t+uu_x + \epsilon^2 u_{xxx})_x + u_{yy}=0,
\end{equation} where $\epsilon>0$ is a parameter representing the strength of dispersion. This version \eqref{eq:kp} with $+u_{yy}$ is also known as KPII \cite{kadomtsev_stability_1970},
a completely integrable equation \cite{ablowitz_solitons_1991} that arises in surface water waves \cite{ablowitz_baldwin_2012,ablowitz_evolution_1979}, internal water waves \cite{grimshaw_evolution_1981,grimshaw_internal_2018}, and ion-acoustic waves in plasma \cite{infeld_2001}. It admits a two-parameter family of stable line soliton solutions,
\begin{equation}\label{eq:soli}
        u(x,y,t) =  a\sech^2\left(\sqrt{\frac{a}{12}}\frac{x+qy-ct}{\epsilon}\right), \qquad c=\frac{a}{3}+ q^2 ,
\end{equation} where $a$ is the soliton peak amplitude, $q=\tan{\varphi}$, $\varphi$ being the counterclockwise angle between the soliton and the $y$-axis (see figure~\ref{fig:soli_solns}$(a)$), and $\epsilon$ is proportional to the soliton's width in the direction of propagation. The KP equation \eqref{eq:kp} is a two-dimensional generalization of the Korteweg-deVries (KdV) equation and therefore a natural framework for the study of line solitons; when $q=0$ in \eqref{eq:soli}, the well-known KdV soliton is recovered. 
\par The KP equation \eqref{eq:kp} also admits a three-parameter family of Y soliton solutions that consist of a resonant triad of line solitons stably joined together (see figure~\ref{fig:soli_solns}$(b)$) and represent the long-time asymptotic state for line solitons after Mach reflection. Objects of much study, Y-shaped configurations have been observed in numerics \cite{funakoshi_1980,tanaka_mach_1993,tsuji_oblique_2007,biondini_soliton_2009,kodama_2009,li_mach_2011,kao_numerical_2012,chakravarty_numerics_2017,knowles_yeh_2019}, physical experiments \cite{russell_1845,gilmore_1950,perroud_solitary_1957,henderson_1975,melville_mach_1980,li_mach_2011,tanaka_1993,yeh_2010,kodama_yeh_2016}, and field observations \cite{ablowitz_baldwin_2012,wang_internal_observed,andaman_sea_2017}. KP Y solitons were first described mathematically by John Miles \cite{Miles1977oblique,Miles1977resonant}, whose results were generalized by the discovery of more complex KP exact solutions that contain Y-shaped substructures \cite{hirota_2004,kodama_2004,biondini_2007,biondini_soliton_2009,chakravarty_kodama_2009,kodama_kp_2010}. Numerical studies have also identified approximately resonant Y-shapes arising in other similar equations \cite{chang_mKP_2019,numerics_2001,oikawa_2006,matsuno_1998}. 

\begin{figure}
    \centering
    \includegraphics[scale=.25,trim={0 5.4cm 3cm 3cm}]{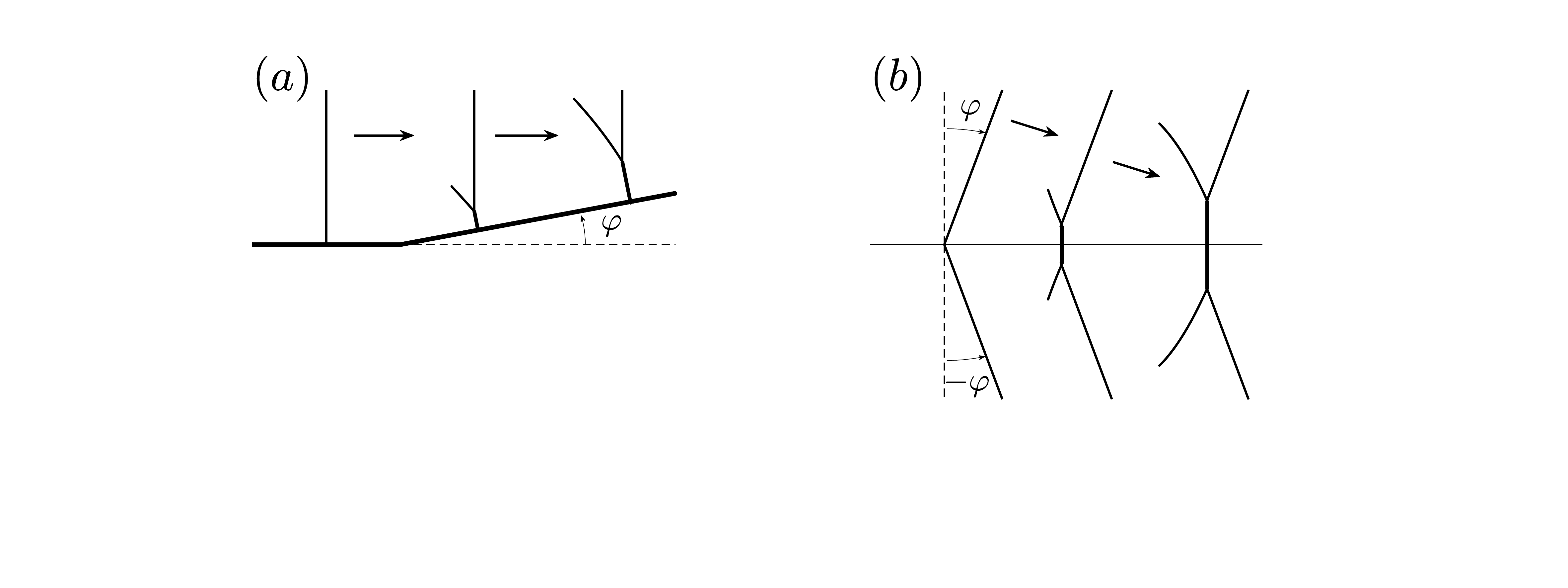}
    \caption{$(a)$ Mach reflection of a soliton or shock impinging on a compressive corner with sufficiently small angle $\varphi$. $(b)$ Mach reflection also occurs for solitons with an initial inward V-shape, which under the KP equation \eqref{eq:kp} is identical to the left figure for the purposes of analysis.}
    \label{fig:mach_ref_schematic}
\end{figure}
\begin{figure}
    \centering
    \includegraphics[scale=.25,trim={0 1.5cm 0 0}]{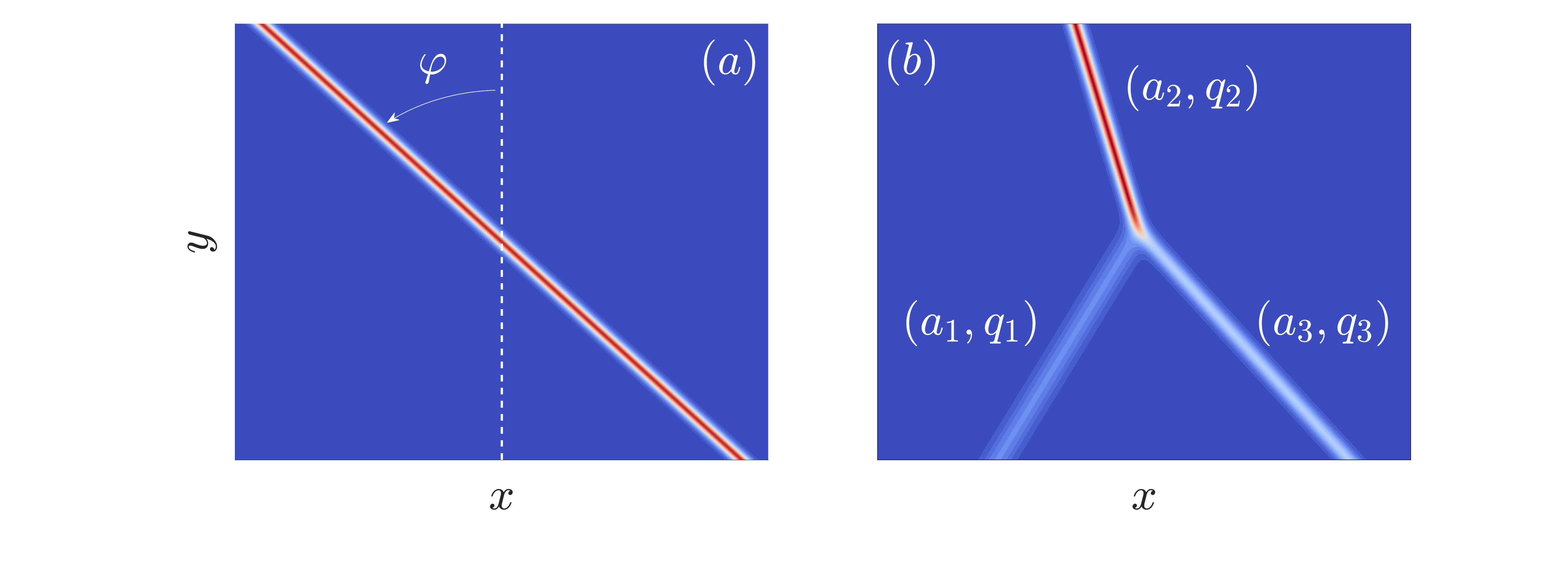}
    \caption{$(a)$ A line soliton exact solution \eqref{eq:soli} to the KP equation \eqref{eq:kp}, where $a$ represents the max amplitude, $q=\tan{\varphi}$, and the width is proportional to $\epsilon$. $(b)$ A resonant Y soliton solution to the KP equation \eqref{eq:kp} with physical variables $(a_i,q_i)$ of the $i^{\rm th}$ soliton labeled. The physical variables are determined by three phase variables according to the conditions \eqref{eq:y_soli_par}. This Y soliton corresponds to the local behaviour of the bottom half of figure~\ref{fig:mach_ref_schematic}$(b)$.}
    \label{fig:soli_solns}
\end{figure}
\par In this paper, we study the Y soliton in the zero dispersion (i.e. $\epsilon \to 0$) limit of the KP equation \eqref{eq:kp}. In this limit, which is equivalent to the limit of long time $t$ and large spatial $(x,y)$ scales for some kinds of data, line solitons \eqref{eq:soli} vanish in width to become localized to a line in the $x$-$y$ plane. When a line soliton is modulated, its parameters $(a,q)$ in the zero dispersion limit evolve according to Whitham modulation equations \cite{whitham,dsw_review_2016,lax_levermore_1983_I,lax_levermore_1983_II,lax_levermore_1983_III}. Recent advances in Whitham modulation theory for the KP equation \eqref{eq:kp} \cite{ablowitz_2017_whitham,ablowitz_2018_whitham,grava_numerical_2018} have proven  to be effective for studying the asymptotic evolution of line solitons interacting due to various initial conditions \cite{biondini2019integrability,ryskamp_2020,ryskamp2020_mf}, some of the first applications of Whitham theory to a governing equation in two spatial and one temporal dimension. 
\par In this work, we obtain two main results by examining the zero dispersion limit of line solitons and Y solitons. The first result is that, in this limit, all continuous line soliton modulations become essentially one-dimensional in the $x$-$y$ plane. We demonstrate that $a$ and $q$ are constrained so that if one knows the parameters' evolution in one spatial dimension, their evolution in the other dimension can be determined. This result will allow us to reduce the (2+1)-dimensional soliton modulation equations to a (1+1)-dimensional system.
\par The second and primary result of this paper is that the KP Y soliton can be interpreted as a shock solution to the soliton modulation equations. In the zero dispersion limit, the three-way intersection of line solitons in the centre of the Y soliton reduces to a single point of discontinuity in the modulation parameters. The velocity of this vertex and relationships between the three solitons' parameters can then be determined using a generalization of Rankine-Hugoniot jump conditions for a family of conservation laws of the modulation equations. In order to account for the multivalued nature of the Y soliton while using (1+1)-dimensional modulation conservation laws,  we introduce \emph{modified Rankine-Hugoniot jump conditions}, where the densities and fluxes of the two smaller legs of the Y soliton are summed. Notably, we show that the family of modulation conservation laws for which the Y soliton vertex satisfies these modified Rankine-Hugoniot jump conditions is \emph{infinite}, contrary to the typical setting of hyperbolic conservation laws. 

\par Our results imply that the Y soliton can be understood as a modified \emph{Whitham shock} \cite{sprenger_whitham_shocks}, a recent discovery defined as a discontinuity in the modulation equations that satisfies Rankine-Hugoniot jump conditions and corresponds to a travelling wave solution to the full equation. Whereas in classical shock theory, a shock is obtained by taking a zero viscosity limit of a travelling wave solution to a dissipative equation, a Whitham shock is obtained by taking the zero dispersion limit of a travelling wave in a dispersive equation.
\par The modulation theory approach features a number of additional novel elements. A physically inspired way to study resonant soliton interactions is by formulating non-soliton initial value problems. The descriptions of the full dynamics of such problems have been limited to numerical simulations, because analytical methods have focused on exact multi-soliton solutions to the KP equation \eqref{eq:kp} \cite{kodama_2009,chakravarty_kodama_2009,yeh_2010,chakravarty_2013,kodama_book,kodama_yeh_2016,chakravarty_numerics_2017}. The modulation approach has the advantage of being able to describe the dynamical evolution of non-soliton initial conditions, not only the long-time, locally steady state. In this paper, we demonstrate this advantage by analytically solving for the evolution of V-shaped initial conditions, i.e. the Mach reflection problem \cite{oikawa_2006,kodama_2009,yeh_2010,kodama_yeh_2016,kodama_book,chakravarty_numerics_2017}. In addition, to the best of our knowledge, this work is the first fully two-dimensional application of Whitham modulation theory, which typically has been restricted to one-dimensional equations or reductions. Consequently, this method holds promise for studying soliton resonance or near resonance in equations where multi-soliton solutions are not known.
\par The remainder of the paper is organized as follows. In section~\ref{sec:miles_res}, we summarize the exact solution forms of the KP line soliton and Y soliton solutions, and we review their zero dispersion limits. Next, in section~\ref{sec:ave_cons_laws}, the derivation of an infinite number of soliton modulation equations in conservation form is presented. In section~\ref{sec:shock_soln}, we show how the Y soliton satisfies (modified) jump conditions for the infinite modulation conservation laws. These results are applied to the classical Mach reflection problem in section~\ref{sec:v_shape}, followed by a discussion in section~\ref{sec:concl}. Appendix~\ref{sec:appendix} contains a proof relevant to our results. In section \ref{sec:v_shape}, our analysis is quantitatively supported by a pseudospectral numerical solver based on \cite{kao_numerical_2012}, the details of which are given in \cite{ryskamp_2020}.
\section{Exact soliton solutions to the KP equation}
\label{sec:miles_res}
\subsection{Exact soliton solutions}
A large class of multi-soliton solutions of the KP equation \eqref{eq:kp}
can be expressed as \cite{kodama_book,Biondini_2006}:
\begin{subequations}
\begin{equation}
u(x,y,t) = 12\, \frac{\partial^2}{\partial x^2}[\log\tau(x,y,t)]\,, 
\label{e:ugeneral}
\end{equation}
where the tau function $\tau(x,y,t)$, 
is given by the Wronskian of exponentials
\label{eq:gen_soln}
\begin{equation}
\tau(x,y,t) = \Wr(f_1,\dots,f_N)\,, \qquad f_n(x,y,t) = \sum_{m=1}^M A_{n,m} e^{\theta_m(x,y,t)}\,,
\end{equation}  
and the phases $\theta_1,\dots,\theta_M$ are 
\begin{equation}
\label{eq:gen_theta}
\theta_m(x,y,t) = \frac{1}{\sqrt{12}\epsilon} \left( k_m x + \frac{k_m^2}{2} y -\frac{k_m^3}{3} t + \theta_{m,0}\right)\,.
\end{equation}
\end{subequations}
The solution \eqref{eq:gen_soln} is uniquely determined by the phase parameters $k_1,\dots,k_M$ and the coefficient matrix $A = (A_{n,m})$,
plus the translation constants $\theta_{1,0},\dots,\theta_{M,0}$.
Without loss of generality, one can take the phase parameters to be strictly ordered so that $k_1<\cdots<k_M$.

Generically,
the above representation produces a solution with exactly $N$ asymptotic line solitons as $y\to\infty$
and $M-N$ asymptotic line solitons as $y\to-\infty$ \cite{Biondini_2006}.
The amplitude and slope of each asymptotic line soliton are completely determined by the
phase parameters $k_1,\dots,k_M$.
In the simplest nontrivial case, obtained when $N =1$ and $M=2$, one recovers the soliton solution~\eqref{eq:soli}, shown in figure~\ref{fig:soli_solns}$(a)$.
It is convenient to label the two phase parameters as $k_-$, $k_+$, $k_- < k_+$ in this case.
The amplitude and slope parameters $a$ and $q$ are given by 
\begin{equation}
a = \frac{1}{4} (k_+ - k_-)^2\,,\qquad
q = \frac{1}{2}(k_+ + k_-)\,.
\label{e:directmap}
\end{equation}
The inverse map to \eqref{e:directmap} is
\begin{equation}
k_- =  q - \sqrt{a}\,,\qquad
k_+ = q + \sqrt{a}\,.
\label{e:inversemap}
\end{equation}
\par The next simplest case is obtained when $N=1$ and $M=3$ and represents the resonant soliton solution \cite{Biondini_2003}, shown in figure~\ref{fig:soli_solns}$(b)$. As $y \to -\infty$ there are two solitons with phase parameters $(k_1,k_2)$ and $(k_2,k_3)$, respectively. Conversely, as $y \to \infty$ one observes only one soliton specified by the phase parameters $(k_1,k_3)$. Consequently, the Y soliton is determined (up to translation constants) by three phase parameters $k_1$, $k_2$, and $k_3$. For leg $j$ of the Y soliton, $j \in \{1,2,3\}$, counting clockwise from the bottom left (see figure~\ref{fig:soli_solns}$(b)$), the physical variables $(a_i,q_i)$ are determined by the phase parameters as
\begin{align}
    \label{eq:y_soli_par}
    \begin{split}
\begin{aligned}
    & q_1 = \frac{1}{2}(k_2+k_1),  \qquad && q_2 = \frac{1}{2}(k_3+k_1), \qquad && q_3 = \frac{1}{2}(k_3+k_2),\\
    & a_1 = \frac{1}{4}(k_2-k_1)^2, \qquad 
     \qquad && a_2 =\frac{1}{4}(k_3-k_1)^2,\qquad
     \qquad && a_3 =\frac{1}{4}(k_3-k_2)^2.
     \end{aligned}\end{split}
\end{align}Due to the ordering $k_1<k_2<k_3$, one can see that the top leg (leg 2, also known as the Mach stem) is also the largest amplitude soliton. The velocity $(s_x,s_y)$ of the Y soliton vertex is obtained by assuming a traveling wave solution satisfying $\theta_1=\theta_2=\theta_3$ \eqref{eq:gen_theta}, which yields
\begin{equation}
    \label{eq:exact_spd}
    (s_x,s_y) = \left(-\frac{1}{3}(k_1 k_2 +k_1 k_3 + k_2 k_3), \frac{2}{3}(k_1+k_2+k_3)  \right).
\end{equation}
Note that by applying the symmetry transformation $y\to-y$, one can apply the same analysis to a Y soliton with two legs as $y \to \infty$ and one leg as $y \to -\infty$. The expressions for the velocities \eqref{eq:exact_spd} remain unchanged. Although the Y soliton is a special exact solution to the KP equation \eqref{eq:kp}, Y-shaped interactions are ubiquitous as building blocks for higher-order solutions to the KP equation \cite{kodama_book}. 
 \par Since the KP equation remains unchanged under the pseudorotation transformation of the form \cite{ablowitz_2017_whitham,finkel_segur_1985},
 \begin{equation}
 \label{eq:pseudorotate}
  q' = Q + q, \qquad x' = x + Qy - Q^2 t, \qquad y' = y-2Qt, \qquad  Q \in \mathds{R},
 \end{equation}
 one can determine $s_x$ in terms of the remaining Y soliton parameters. For a general Y soliton with parameters \eqref{eq:y_soli_par}, we consider the rotated image of the Y soliton with parameters denoted by $'$ and a vertical stem $q_2'  =0$. The vertical velocity of the Y soliton image from \eqref{eq:exact_spd} is $s_y'=2 k_2'/3$ and its horizontal velocity must be the velocity of the stem $s_x'=a_2/3$ (see \eqref{eq:soli}). If we now rotate the image Y soliton to the initial Y soliton via a rotation of $Q=q_2$, we can use the symmetries \eqref{eq:pseudorotate} to obtain the horizontal velocity that is equivalent to \eqref{eq:exact_spd},
 \begin{equation}
\label{eq:rotate_spd}
s_x = \frac{a_2}{3}+q_2^2 - q_2 s_y . 
 \end{equation}

 \subsection{Zero dispersion limit of KP solitons}
 We will study the Y soliton in the zero dispersion ($\epsilon \to 0$) limit. The equation for the potential $\phi$ of a single line soliton \eqref{eq:soli} where $u=\epsilon \phi_x$ is  \begin{equation}
    \label{eq:potential}
        \phi(x,y,t) = \sqrt{12}q+\sqrt{12a}\, \tanh\left(\sqrt{\frac{a}{12}}\frac{x+qy-ct}{\epsilon}\right), \qquad c = \frac{a}{3}+q^2,
\end{equation}
where the first term is a constant of integration that is well-defined by using the tau function form \eqref{eq:gen_soln}. The zero dispersion parameter $\epsilon$ represents the transition width of the $\tanh$ function, which corresponds to the width of the soliton \eqref{eq:soli}. The limit of the potential \eqref{eq:potential} as $\epsilon \to 0$  becomes a discontinuous step,
\begin{equation}
    \label{eq:heaviside}
   \lim_{\epsilon \to 0} \phi(x,y,t) =  \sqrt{12}\left[(q-\sqrt{a})+2\sqrt{a}\, H\left(x+qy-ct\right)\right],
\end{equation} where $H$ is the Heaviside (step) function. Thus, the corresponding soliton solution to the KP equation in the zero dispersion limit, which is the scaled $x$-derivative of \eqref{eq:heaviside}, is a zero-width impulse with parameters $a$ and $q$, where $a$ is the impulse strength and $q$ is a measure of the slope of the line in the $x$-$y$ plane. In other words, the soliton is localized along the line $x+q y = c t$. This interpretation of a line soliton is similar to what has been called a \emph{soliton graph} \cite{kodama_book}, with the zero dispersion limit corresponding to the so-called ``tropical" limit.
\par Consequently, in the zero dispersion limit, the three line solitons which compose a Y soliton are localized to three semi-infinite lines with a single point of intersection. This intersection point is a travelling discontinuity between all three zero-width, constant line soliton impulses. It is the behaviour of this point that we will interpret with Whitham modulation theory. 

\section{Soliton modulation theory}
\label{sec:ave_cons_laws}
\subsection{Derivation of one-dimensional soliton modulation equations}
\par In this section, we will derive one-dimensional modulation conservation laws for the parameters $(a,q)$ of the line soliton solution \eqref{eq:soli}. These equations describe the parameters' evolution in the zero dispersion limit where the line soliton reduces to the scaled derivative of \eqref{eq:heaviside}, an impulse along a curve in the $x$-$y$ plane. Fully two-dimensional soliton modulation equations have been derived in prior work as \cite{biondini2019integrability,grava_numerical_2018}
\begin{equation}
  \label{eq:full_mod_syst}
  \begin{bmatrix}
    a \\ q
  \end{bmatrix}_t +
  \begin{bmatrix}
    \frac{1}{3}a - q^2 & -\frac{4}{3} a q \\
    -\frac{1}{3}q & \frac{1}{3}a - q^2
  \end{bmatrix}
  \begin{bmatrix}
    a \\ q
  \end{bmatrix}_x + 
  \begin{bmatrix}
    2q & \frac{4}{3}a \\
    \frac{1}{3} & 2q
  \end{bmatrix}
  \begin{bmatrix}
    a \\ q
  \end{bmatrix}_y
  = 0 .
\end{equation} We will use the following theorem, proved in Appendix~\ref{sec:appendix}, to justify considering a one-dimensional reduction of \eqref{eq:full_mod_syst}.
\begin{theorem}
\label{def:thm}
If $(a_x,q_x)$ and $(a_y,q_y)$ are well-defined and continuous, the modulation equations \eqref{eq:full_mod_syst} for the parameters $(a,q)$ of a KP line soliton \eqref{eq:soli} can always be reduced to an equivalent (1+1)-dimensional system of equations.
\end{theorem}
Theorem~\ref{def:thm} can be intuitively understood by the fact that since a line soliton is localized to a curve in the zero dispersion limit, modulations in one spatial dimension can be locally mapped to modulations in the other spatial dimension using the line equation $x+qy=ct$. Consequently, in most of this paper, without loss of generality we assume $x$-independence of the modulation parameters. In section~\ref{sec:y_indep}, we will show how to map these results to $y$-independent parameters. As an aside, a natural approach for line soliton modulations would be to use an arc length parametrization, similar to equations obtained for modulations of line dark solitons of the nonlinear Schr\"{o}dinger equation \cite{mironov}.
\par In light of the above discussion, using the dispersion parameter $\epsilon$, we assume the $x$-independent asymptotic expansion
\begin{equation}
    \label{eq:exact_expans}
    u(x,y,t; \epsilon) = u_0(\theta,y,t) + \epsilon u_1(\theta,y,t) + \cdots,
\end{equation}
with the leading order modulated soliton ansatz
\begin{equation}\label{eq:x_indep_soli}
    \begin{split}
        u_0(\theta,y,t; \epsilon) &=  a(y,t)\sech^2\left(\sqrt{\frac{a(y,t)}{12}}\frac{\theta}{\epsilon}\right), \\
        \theta_{x} &= 1,\quad \theta_{y}=q(y,t), \quad \theta_{t} = -\left(\frac{a(y,t)}{3}+q^2(y,t) \right),
    \end{split}
\end{equation}
and where every $u_i$ and their derivatives in \eqref{eq:exact_expans} vanish as $\theta \to \pm \infty$. Due to division by $\epsilon$ in \eqref{eq:x_indep_soli}, $\theta/\epsilon$ is a fast variable compared to $y$ and $t$. We require the compatibility conditions $\theta_{xy}=\theta_{yx}$, $\theta_{xt}=\theta_{tx}$, and $\theta_{yt}=\theta_{ty}$. The first two compatibility conditions are satisfied by the assumption that $a$ and $q$ are $x$-independent, i.e. $a_x \equiv q_x \equiv 0$. The third condition $\theta_{yt}=\theta_{ty}$ yields the equation
\begin{equation}
    \label{eq:11}
\theta_{yt}-\theta_{ty}=q_t + \left(\frac{a}{3}+q^2 \right)_y=0,
\end{equation}
which immediately gives us a modulation equation for $q$ in conservation form. This equation \eqref{eq:11} is known as conservation of waves.
\par We will now average the conservation laws of the KP equation over the fast variable $\theta$ to obtain a number of additional modulation conservation laws for the soliton amplitude \cite{whitham_1965_cons,whitham,dsw_review_2016}. Consider a general conservation law
\begin{equation}
    \label{eq:kp_cons_law_gen}
    \frac{\partial}{\partial t}F +\frac{\partial}{\partial x}G + \frac{\partial}{\partial y}H = 0,
\end{equation}
where $F,G,H$ depend on the derivatives of the potential $\phi$ such that $u=\epsilon \phi_x$. Based on the modulated soliton form \eqref{eq:x_indep_soli}, we can expand the derivatives as
\begin{equation}
    \label{eq:deriv_expand}
    \frac{\partial}{\partial t} \mapsto  -\frac{c}{\epsilon} \frac{\partial}{\partial \theta} +  \frac{\partial}{\partial t}, \qquad  \frac{\partial}{\partial x} \mapsto  \frac{1}{\epsilon} \frac{\partial}{\partial \theta}, \qquad \frac{\partial}{\partial y} \mapsto  \frac{q}{\epsilon} \frac{\partial}{\partial \theta} +  \frac{\partial}{\partial y}.
\end{equation}
After inserting the derivatives \eqref{eq:deriv_expand} and the expansion \eqref{eq:exact_expans} into \eqref{eq:kp_cons_law_gen}, we then apply the integral operator $\bar{F}=\int_{-\infty}^{\infty} F(u) \, \mathrm{d}\theta$ to each element in the equation. 
All total derivatives with respect to $\theta$ disappear, since we assumed that every $u_i$ and their derivatives vanish as $\theta \to \pm \infty$. We are therefore left with
\begin{equation}
    \int_{-\infty}^\infty\left( \frac{\partial F}{\partial t} + \frac{\partial G}{\partial x}+\frac{\partial H}{\partial y} \right) \mathrm{d}\theta= \frac{\partial \overline{F}}{\partial t} + \frac{\partial \overline{H}}{\partial y} + \mathcal{O}(\epsilon).
\end{equation}
\par As a consequence of integrability, the KP equation \eqref{eq:kp} admits an infinite number of conservation laws of the form \eqref{eq:kp_cons_law_gen}. We list three of them here \cite{anco_2018,anco_2021}, given in terms of the potential $\phi$ where $u=\epsilon \phi_x$:
\begin{subequations}
\label{eq:conserv_laws}
\begin{align}
\label{eq:cons_law_1}
    \left(\phi_x^2 \right)_t + \left(2\epsilon^2 (\phi_x\phi_{xxx}- \phi_{xx}^2)+\frac{2\epsilon}{3} \phi_x^3 -\phi_y^2\right)_x + \left(2\phi_x \phi_y \right)_y &= 0, \\
\label{eq:cons_law_2}
 \begin{split}   \left(\phi_x \phi_y \right)_t + \left(2\epsilon^2 (\phi_y\phi_{xxx}-\phi_{xx}\phi_{xy})+\epsilon \phi_y\phi_x^2+\phi_t\phi_y \right)_x + \\ + \left(\epsilon^2 \phi_{xx}^2+\phi_y^2-\frac{\epsilon}{3}\phi_x^3-\phi_t\phi_x \right)_y &= 0, \end{split} \\
\label{eq:cons_law_3}
 \left(\frac{\epsilon \phi_x^3}{3}-\epsilon^2 \phi_{xx}^2+\phi_y^2 \right)_t +\left(2\epsilon^2(\phi_{xt}\phi_{xx}-\phi_t \phi_{xxx})-\epsilon \phi_t \phi_{x}^2-\phi_t^2\right)_x - \left(2\phi_t \phi_y \right)_y&=0.
\end{align}\end{subequations}
These equations correspond to conservation of momentum in $x$ \eqref{eq:cons_law_1}, momentum in $y$ \eqref{eq:cons_law_2}, and energy \eqref{eq:cons_law_3} in shallow water applications. After transforming the derivatives according to \eqref{eq:deriv_expand}, inserting the expansion \eqref{eq:exact_expans}, and applying the integral operator, the zero dispersion limit of \eqref{eq:conserv_laws} becomes
\begin{subequations} 
\label{eq:conserv_laws_mod}
\begin{align}
    \label{eq:conserv_laws_mod1}
    \left(\overline{u_0^2} \right)_t  + \left(2q\overline{u_0^2} \right)_y &= 0 \\
    \label{eq:conserv_laws_mod2}
    \left(q\overline{ u_0^2 }\right)_t  + \left(\overline{u_{0\theta}^2}-\frac{\overline{u_0^3}}{3}+(q^2+c)\overline{u_0^2} \right)_y &= 0 \\
    \label{eq:conserv_laws_mod3}
    \left(-\overline{u_{0\theta}^2}+\frac{\overline{u_0^3}}{3}+q^2 \overline{u_0^2} \right)_t + \left(2qc\overline{ u_0^2} \right)_y &= 0.
\end{align}
\end{subequations}
Upon substituting \eqref{eq:x_indep_soli}, the integrals in \eqref{eq:conserv_laws_mod} can be calculated explicitly to yield the modulation conservation laws
\begin{subequations}
\label{eq:ave_cons}
\begin{align}
    \label{eq:ave_cons_1}
    \left( a^{3/2} \right)_t + \left( 2qa^{3/2}\right)_y &= 0, \\ 
     \label{eq:ave_cons_2}
     \left(q a^{3/2} \right)_t + \left(\frac{2}{15}a^{5/2}+2q^2a^{3/2} \right)_y &= 0, \\ 
    \label{eq:ave_cons_3}
    \left( \frac{a^{5/2}}{5}+q^2 a^{3/2} \right)_t + \left( \frac{2}{3}qa^{5/2}+2q^3 a^{3/2}\right)_y &= 0.
\end{align}\end{subequations}
The expression \eqref{eq:ave_cons_1} corresponds to conservation of wave action for soliton modulations. Expanding the derivatives in each of \eqref{eq:ave_cons} yields equivalent modulation equations for $a$ and $q$. Together with \eqref{eq:11}, the $x$-independent modulation equations for $a$ and $q$ can then be written in matrix form as the $x$-independent reduction of \eqref{eq:full_mod_syst},
\begin{equation}
  \label{eq:mod_syst}
  \begin{bmatrix}
    a \\ q
  \end{bmatrix}_t  + 
  \begin{bmatrix}
    2q & \frac{4}{3}a \\
    \frac{1}{3} & 2q
  \end{bmatrix}
  \begin{bmatrix}
    a \\ q
  \end{bmatrix}_y
  = 0 .
\end{equation}

Importantly, all three modulation conservation laws \eqref{eq:ave_cons} are consistent with the system \eqref{eq:mod_syst}. This system \eqref{eq:mod_syst} has previously been derived using multiple scales \cite{neu_singular_2015,ryskamp_2020} and averaged Lagrangian methods \cite{grava_numerical_2018}. The KP equation has an additional conservation law, conservation of mass (with density $u=\epsilon\phi_x$), but we do not consider this law as it yields results inconsistent with \eqref{eq:mod_syst} and \eqref{eq:ave_cons}. Note also that the modulation system \eqref{eq:mod_syst} and therefore all modulation conservation laws \cite{herreman_2010} possess the scaling symmetry
\begin{equation}
    \label{eq:scale_sym}
    (y,t,a,q) \longrightarrow (\lambda^2 y', \lambda^3 t',\lambda^2 a', \lambda q').
\end{equation}
\par The infinite number of KP conservation laws suggests that, in principle, one could continue this averaging process to find an infinite number of modulation conservation laws. In the next section we will use a more direct approach to obtain an infinite number of modulation conservation laws.

\subsection{Derivation of generalized conservation laws}
\label{sec:mod_cons_laws}
In this section we generate an infinite family of conservation laws consistent with the modulation system \eqref{eq:mod_syst}, following an approach developed by G.B. Whitham \cite{whitham}. A general soliton modulation conservation law in $y$, like those given in \eqref{eq:ave_cons}, has the form
\begin{equation}
    \label{eq:gen_cons_law}
    \frac{\partial f}{\partial t}+\frac{\partial h}{\partial y} = 0,
\end{equation}
where $f=f(a,q)$ and $h=h(a,q)$. Applying the chain rule and inserting the evolution equations \eqref{eq:mod_syst} yields
  \begin{equation}
    \label{eq:ins_gen_cons}
    \left( -2qf_a -\frac{1}{3}f_q+h_a\right)a_y + \left(-\frac{4}{3}af_a-2qf_q + h_q \right)q_y =0.
\end{equation}
Since \eqref{eq:ins_gen_cons} must be true for all values of $q_y$ and $a_y$, the terms in the parentheses must both equal zero, leading to two compatibility equations for $f$ and $h$.
 Let us consider a general density of the form
\begin{equation}
    \label{eq:gen_density}
    f(a,q) = \sum^n_{i=1} c_i a^{i+\frac{1}{2}}q^{2(n-i)}, \qquad c_1 = 1.
\end{equation}
This density ansatz is motivated by analogy with \eqref{eq:ave_cons_1} (which corresponds to \eqref{eq:gen_density} with $n=1$) and \eqref{eq:ave_cons_3} (which corresponds to \eqref{eq:gen_density} with $n=2$). In addition, the ansatz \eqref{eq:gen_density} satisfies the scaling symmetries \eqref{eq:scale_sym} of the modulation system \eqref{eq:mod_syst}. Taking the partial derivatives of \eqref{eq:gen_density} and inserting into the compatibility conditions from \eqref{eq:ins_gen_cons} yields
\begin{subequations}
\label{eq:gen_g}
\begin{align}
    \label{eq:gen_g_a}
    h_a &= \sum^n_{i=1} \left[(2i+1)q+\frac{2}{3}(n-i) \right]c_i a^{i+\frac{1}{2}}q^{2(n-i+\frac{1}{2})},\\
    \label{eq:gen_g_b}
    h_q &= \sum^n_{i=1} \left[\frac{2}{3}(2i+1)+4(n-i) \right] c_i a^{i+\frac{1}{2}}q^{2(n-i)}.
\end{align}\end{subequations}
Requiring equality of coefficients for $h_{aq}=h_{qa}$ gives a recursion relation for $c_i$ in terms of $c_{i-1}$, which can be rewritten explicitly as:
\begin{equation}
    \label{eq:c_coeffs}
     c_i =  \frac{3i}{(n-i)(2n)(2n-1)}{2n \choose 2i+1}, \quad i = 1,\dots, n-1, \qquad c_n = \frac{3}{(2n+1)(2n-1)}.
\end{equation}
By integrating $h_q$ \eqref{eq:gen_g_b} with respect to $q$ we obtain the flux as
\begin{equation}
    \label{eq:gen_flx}
    h(a,q)= \sum^n_{i=1} d_i a^{i+\frac{1}{2}}q^{2(n-i)+1},
\end{equation}
where the coefficients $d_i$ are defined as
\begin{equation}
    \label{eq:d_coeffs}
    d_i = \left(2-\frac{4(i-1)}{3(2i-2n-1)} \right)c_i, \qquad i=1,\dots,n.
\end{equation} 
Up to an arbitrary additive constant, the same result \eqref{eq:gen_flx} is also obtained by integrating \eqref{eq:gen_g_a}. Since the constant is arbitrary, we set it to zero. For $n=1,2,\dots$, the equations \eqref{eq:gen_cons_law}, \eqref{eq:gen_density}, \eqref{eq:c_coeffs}, \eqref{eq:gen_flx}, and \eqref{eq:d_coeffs} together form an infinite family of modulation conservation laws that are consistent with the modulation system \eqref{eq:mod_syst}. This family includes \eqref{eq:ave_cons_1} and \eqref{eq:ave_cons_3} as the $n=1$ and $n=2$ cases, respectively. The $n=3$ case is
\begin{equation}
\label{eq:n_3_conserv}
        \left(\frac{3}{35} a^{7/2}+ \frac{6}{5} q^2 a^{5/2}+q^4 a^{3/2} \right)_t + \left( \frac{2}{5}qa^{7/2}+\frac{44}{15}q^3a^{5/2}+2q^5a^{3/2} \right)_y = 0.
\end{equation}

We can use a different density ansatz to obtain a second infinite family of modulation conservation laws. This family includes \eqref{eq:ave_cons_2} as the $n=1$ case. The process is the same as above, so we simply state the results. The density $\tilde{f}(a,q)$ takes the form
\begin{subequations}
\label{eq:density_2}
\begin{align}
    \tilde{f}(a,q) &= \sum^n_{i=1} \tilde{c}_i a^{i+\frac{1}{2}}q^{2(n-i)+1},\qquad \quad
    \tilde{c}_1 = 1, \\ \tilde{c}_i &= \frac{(2i-2n-3)(i-n-1)}{(i-1)(2i+1)}\tilde{c}_{i-1},\qquad i = 2,\dots,n.
\end{align}\end{subequations}
The corresponding flux $\tilde{h}(a,q)$ can be found to be
\begin{subequations}
\label{eq:flux_2}
\begin{align}
    \tilde{h}(a,q) &= \sum^{n+1}_{i=1} \tilde{d}_i a^{i+\frac{1}{2}}q^{2(n-i)+2}, \qquad \tilde{d}_1=2, \qquad \tilde{d}_{n+1}=\frac{2}{9+6n}\tilde{c}_n,\\ \tilde{d}_{i}&=2\tilde{c}_{i}+\frac{2n-2i+1}{3(i+3/2)}\tilde{c}_{i-1}, \qquad i = 2,\dots,n.
\end{align}\end{subequations}
\section{Shock solutions of the soliton modulation equations}
\label{sec:shock_soln}
In this section we show that the Y soliton can be understood as a discontinuous, weak solution to all the modulation conservation laws given in section~\ref{sec:ave_cons_laws}. In particular, the velocity of the vertex $s_y$ of the Y soliton and each legs' parameters satisfy Rankine-Hugoniot conditions for the conservation of waves equation \eqref{eq:11} and \emph{modified} Rankine-Hugoniot jump conditions for all other modulation conservation laws given above. We also summarize some properties of the Y soliton shock.

\subsection{Shock solution to general KP conservation laws}
\label{sec:traveling_waves}
We can use jump conditions arising from any averaged KP conservation law to determine the correct velocity of the Y soliton. However, it is evident from figure~\ref{fig:soli_solns}$(b)$ that the modulation equations for $a(y,t)$ and $q(y,t)$ are \emph{multivalued} in $y$; for example, for $y<z(t)$, $\Dot{z}=s_y$, $a(y,t)$ takes on two different constant values: $a(y,t)=a_1$ and $a(y,t)=a_3$.  Consequently, the calculation of the jump condition must be modified from traditional Rankine-Hugoniot jump conditions, which are defined using the difference operator $\llbracket u \rrbracket = u_+-u_-$ between two states $u_\pm$. Instead, we define a modified difference operator:
\begin{definition}
\label{def:mod_diff_oper}
Consider a multivalued function $u(y)$ with a discontinuity at $y=z(t)$ and
\begin{equation*}
   \lim_{y\to z(t)^+} u = \{ u_+^{(1)},\dots,u_+^{(n)} \}, \qquad  \lim_{y\to z(t)^-}u= \{u_-^{(1)},\dots,u_-^{(m)} \}, \qquad \Dot{z}(t) = s_y(t).
\end{equation*} The \emph{multivalued difference operator} $\llbracket \, \rrbracket_{n,m}$ is defined as 
\begin{equation*}
    \llbracket u \rrbracket_{n,m}=\sum_{j=1}^n u_+^{(j)} -\sum_{j=1}^m u_-^{(j)}.
\end{equation*}
\end{definition}
Using the multivalued difference operator, we can define modified Rankine-Hugoniot conditions:
\begin{definition}
\label{def:mod_rh_cond}
For a conservation law $(f(u))_t + (h(u))_y = 0$, where $u(y,t)$ is multivalued and discontinuous as in definition~\ref{def:mod_diff_oper}, the discontinuity with velocity $s_y$ is calculated using \emph{modified Rankine-Hugoniot conditions} as
\begin{equation*}
    s_y \llbracket f(u) \rrbracket_{n,m} = \llbracket h(u) \rrbracket_{n,m}\, .
\end{equation*}
\end{definition}
Since both $a(y,t)$ and $q(y,t)$ are multivalued for either $y<z(t)$ or $y>z(t)$ for the Y soliton, we will apply modified Rankine-Hugoniot jump conditions to the modulation conservation laws derived above. The summation of multivalued densities and fluxes in definition~\ref{def:mod_rh_cond} will be justified below. First, we demonstrate how to perform this calculation using the wave action conservation law \eqref{eq:ave_cons_1}, for which modified Rankine-Hugoniot jump conditions take the form
\begin{equation}
\label{eq:wave_action_y_jump_cond}
 s_y =\frac{\left\llbracket 2qa^{3/2} \right\rrbracket_{n,m}}{\left\llbracket a^{3/2} \right\rrbracket_{n,m} }\,.
 \end{equation}
The multivalued difference operator $\llbracket \, \rrbracket_{n,m}$ is applied by adding the appropriate quantities that share the same half plane in $y$, in this case legs 1 and 3 (see figure~\ref{fig:soli_solns}$(b)$), so that $n=1$, $m=2$, and
 \begin{equation}
 \label{eq:amp_48}
    s_y = \frac{\left(2q_2a_2^{3/2}\right)-\left(2q_1a_1^{3/2}+2q_3 a_3^{3/2}\right)}{\left(a_2^{3/2}\right)-\left(a_1^{3/2}+a_3^{3/2}\right)} = \frac{2}{3}\left(k_1 + k_2 + k_3 \right),
 \end{equation}
 where the simplification from physical to phase variables follows from \eqref{eq:y_soli_par} after some algebra. The result \eqref{eq:amp_48} agrees with the known velocity $s_y$ of the Y soliton vertex in $y$ \eqref{eq:exact_spd}. We remark that modified jump conditions for the zero dispersion limit of the Y soliton applied to any of the modulation conservation equations \eqref{eq:ave_cons} will also correctly determine $s_y$.
\par One can justify the adding of the fluxes and densities in $y$ as follows. Consider any conservation law of the full KP equation of the form \eqref{eq:kp_cons_law_gen}. A general travelling wave solution that is localized in $x$ with velocity $(s_x,s_y)$ has the form
\begin{equation}
    \label{eq:travel_ansatz}
    u = u(\xi,\eta), \qquad \xi = x-s_x t, \qquad \eta = y-s_y t.
\end{equation}
Inserting \eqref{eq:travel_ansatz} into \eqref{eq:kp_cons_law_gen} yields
 \begin{equation}
    \label{eq:88}
     -s_x F_\xi - s_y F_\eta +G_\xi+H_\eta =0.
 \end{equation}
 We now fix $\eta$ and integrate \eqref{eq:88} with respect to $\xi$ from $-\infty$ to $\infty$. This definite integral can be calculated for the first and third terms of \eqref{eq:88} by evaluating $F$ and $G$ at $\xi \to \pm \infty$, where they are zero due to the localized nature of $u$. Thus, we obtain
 \begin{equation*}
\left(\int_{-\infty}^\infty -s_y F(\xi,\eta) + H(\xi,\eta) \, \mathrm{d}\xi \right)_\eta = 0.
\end{equation*}
This expression can be integrated with respect to $\eta$ to yield 
\begin{equation}
\label{eq:91}
-s_y \overline{F}(\eta)+\overline{H}(\eta) = A,
 \end{equation}
 where $A$ is a constant, and the average, denoted e.g. by $\overline{F}$, is calculated by integrating over the real line in the $\xi$-direction. Evaluating \eqref{eq:91} at two different values $\eta^+ \gg 1$ and $-\eta^- \gg 1$ allows us to eliminate the constant $A$, and we can solve for $s_y$ to obtain
 \begin{equation*}
     s_y = \frac{\overline{H}(\eta^+)-\overline{H}(\eta^-)}{\overline{F}(\eta^+)-\overline{F}(\eta^-)},
 \end{equation*}
   which has an identical form to \eqref{eq:wave_action_y_jump_cond}. For the Y soliton given in section~\ref{sec:miles_res}, averaging the flux and density over $\xi$ from $-\infty$ to $\infty$, at some $-\eta^- \gg 1$, is equivalent to adding together the fluxes and densities of the two smaller legs (legs 1 and 3).  Thus, $\overline{F}(\eta^-)$ can be expressed as $\overline{F}_1+\overline{F}_3$, where $\overline{F}_i$ denotes the value of the averaged flux $\overline{F}$ at the $i^{\rm th}$ leg. Similarly, $\overline{H}(\eta^-)=\overline{H}_1+\overline{H}_3$. Our final expression for the velocity in $y$ becomes
 \begin{equation*}
          s_y = \frac{\overline{H}_2-\left( \,\overline{H}_1+\overline{H}_3\right)}{\overline{F}_2-\left(\, \overline{F}_1+\overline{F}_3\right)},
 \end{equation*}
which corresponds to the addition of fluxes and densities in the averaged conservation law in \eqref{eq:amp_48}. In the zero dispersion limit, the Y soliton vertex region shrinks to a point, and we merely require $\eta^+>0$ and $\eta^-<0$. A similar construction holds if the multivalued pair of Y soliton legs occurs for $y>z(t)$ where now $n=2$, $m=1$ in \eqref{eq:wave_action_y_jump_cond}.

\subsection{Shock solutions to families of conservation laws}
In this section, we will show that the Y soliton satisfies modified Rankine-Hugoniot conditions (see definition~\ref{def:mod_rh_cond}) for the infinite family of conservation laws given above. Specifically, we will prove the following:
\begin{proposition}
Consider a modulation conservation law \eqref{eq:gen_cons_law} with density given by \eqref{eq:gen_density} and \eqref{eq:c_coeffs} and flux given by \eqref{eq:gen_flx} and \eqref{eq:d_coeffs}, respectively. Then modified Rankine-Hugoniot jump conditions for the conservation law applied to the Y soliton parameters \eqref{eq:y_soli_par} yield
\begin{equation}
    \label{eq:pf_goal}
    \frac{\left\llbracket h(a,q) \right\rrbracket_{1,2} }{\left\llbracket f(a,q) \right\rrbracket_{1,2} }=\frac{2}{3}(k_1 + k_2 + k_3)=s_y .
\end{equation}
\end{proposition}
\begin{proof}
 We proceed by rewriting the density for a single soliton \eqref{eq:gen_density} in terms of phase variables $s=q+\sqrt{a}$ and $r=q-\sqrt{a}$ \eqref{e:inversemap} as
\begin{equation}
    \label{eq:ri_dens}
    f(s,r) = \sum_{i=1}^n \frac{c_i}{2^{2n+1}}(s - r)^{2i+1} (s+r)^{2(n-i)}.
\end{equation}
The density polynomial \eqref{eq:ri_dens} can be rewritten as
\begin{equation*}
    f(s,r) = (s^{2n+1}-r^{2n+1}) \sum_{i=1}^n \frac{c_i}{2^{2n+1}} + (r^{2n}s-r s^{2n}) \sum_{i=1}^n \frac{c_i}{2^{2n+1}}(4i-2n+1).
\end{equation*}
These terms come from expanding the binomials in \eqref{eq:ri_dens} and distributing for the product. Any term that contains a power of less than $2n$ cancels. Similarly, the flux \eqref{eq:gen_flx} can be rewritten as
\begin{align*}
        h(s,r) &= \sum_{i=1}^n \frac{d_i}{2^{2n+1}}(s - r)^{2i+2} (s+r)^{2(n-i)+1}, \\
        \begin{split} &= (s^{2n+2}-r^{2n+2}) \sum_{i=1}^n \frac{d_i}{2^{2n+2}} + (r^{2n}s-r s^{2n}) \sum_{i=1}^n \frac{d_i}{2^{2n+2}}(4i-2n) + \\ &\quad +  (r^{2n-1}s^2-r^2 s^{2n-1})\sum_{i=1}^n \frac{d_i}{2^{2n+2}}(2n^2-n-8in+8i^2-1). \end{split}
\end{align*}
If we now insert the phase variables $k_i$ for each leg, evaluate $\left\llbracket f \right\rrbracket_{1,2} $, and multiply by the velocity $s_y=\frac{2}{3}(k_1+k_2+k_3)$, we find after some algebra that this product is equal to $\left\llbracket h \right\rrbracket_{1,2} $ if and only if the following equalities hold
\begin{equation}
\label{eq:sums}
    \sum_{i=1}^n \frac{c_i}{3} (4i-2n+1) =\sum_{i=1}^n \frac{d_i}{2}(2i-n) = \sum_{i=1}^n -\frac{d_i}{4}(2n^2-n-8in+8i^2-1).
\end{equation}
Thus, if we can prove \eqref{eq:sums}, we have shown \eqref{eq:pf_goal}. Using \eqref{eq:c_coeffs} and \eqref{eq:d_coeffs}, a calculation demonstrates that \eqref{eq:sums} holds.
\end{proof}
Thus, the parameters of the resonant Y soliton satisfy modified Rankine-Hugoniot jump conditions for an infinite family of conservation laws in $y$. We omit the similar proof for the second family of conservation laws given in \eqref{eq:density_2} and \eqref{eq:flux_2}.
\subsection{Shock solution to conservation of waves}
\label{sec:cons_waves_spd}
In this section, we use \emph{ordinary} Rankine-Hugoniot jump conditions for the modulation equation for $q$ \eqref{eq:11} that arises from conservation of waves. Instead of modified Rankine-Hugoniot jump conditions, we choose any two legs of the Y soliton and insert them into the ordinary jump conditions for \eqref{eq:11}. The justification for using modified Rankine-Hugoniot jump conditions above depended on the fact that the modulation conservation laws were derived from KP conservation laws. In contrast, the wave action conservation law \eqref{eq:11} is a kinematic condition arising from the definition of our modulation ansatz \eqref{eq:x_indep_soli}. A similar kinematic calculation to determine the speed of multisoliton structures was utilized in \cite{ostrovsky_kinematics_2020}. 

Without loss of generality, we consider the jump between legs 1 and 2 (see figure~\ref{fig:soli_solns}$(b)$). Then, ordinary jump conditions ($\llbracket \, \rrbracket \equiv \llbracket \, \rrbracket_{1,1}$) for \eqref{eq:11} become (with simplification using \eqref{eq:y_soli_par})
\begin{align}
    \label{eq:phase_jump_cond_y}
s_y\llbracket q\rrbracket &= \left\llbracket \frac{a}{3}+q^2 \right\rrbracket,\\
s_y &= \frac{\frac{1}{3}(k_2-k_1)^2+(k_2+k_1)^2-\frac{1}{3}(k_3-k_1)^2-(k_3+k_1)^2}{2((k_2+k_1)-(k_3+k_1))}= \frac{2}{3}(k_1+k_2+k_3),
\end{align}
which is the correct velocity $s_y$ \eqref{eq:exact_spd}. One can verify that the jump condition \eqref{eq:phase_jump_cond_y} yields the correct velocities for the choice of \emph{any two legs} from the three that make up the resonant Y soliton. This is true even for the two legs in the multivalued region, i.e. inserting legs 1 and 3 into the jump conditions in $y$ \eqref{eq:phase_jump_cond_y}.
\par
We justify this result as follows. The velocity of the vertex of the Y soliton is calculated directly by setting $\theta_1=\theta_2=\theta_3$, where $\theta_i$ is the non-modulated phase \eqref{eq:soli} of the $i^{\rm th}$ soliton leg. Consider $\theta_i=\theta_j$, which can be rewritten as
\begin{equation*}
     (q_i-q_j )y = \left(\frac{a_i}{3} +q_i^2-\frac{a_j}{3}+q_j^2 \right)t,
\end{equation*}
up to an overall additive phase constant that we have set to zero. Since $y/t=s_y$, the above equation is precisely the jump condition \eqref{eq:phase_jump_cond_y}. Thus, the requirement that $\theta_i=\theta_j$ is equivalent to calculating the shock velocity $s_y$ using ordinary jump conditions. The jump conditions \eqref{eq:phase_jump_cond_y} are essentially kinematic conditions ensuring that the three solitons are travelling together.

\subsection{Shocks in $y$-independent variables}
\label{sec:y_indep}
In this section, we briefly show how the results from the above sections using $x$-independent equations can be repeated for $y$-independent equations, correctly obtaining the velocity $s_x$ of the Y soliton vertex in $x$. In order to transform $x$-independent modulation equations \eqref{eq:mod_syst} to $y$-independent modulation equations, one can use the identification
\begin{equation}
    \label{eq:y_to_x_trans}
    \partial_y \to -q \partial_x, \qquad \partial_t \to \partial_t +c \partial_x,
\end{equation}
where the first expression in \eqref{eq:y_to_x_trans} transforms the derivative with respect to $y$ to a derivative with respect to $x$ (assuming $q\neq 0 $). The second expression in \eqref{eq:y_to_x_trans} reflects that modulations in $x$ must occur at the velocity $c$ of the soliton in $x$.

Assuming a modulation conservation law in $y$ with the general form $f_t+h_y = 0$, one can utilize the transformation \eqref{eq:y_to_x_trans} combined with conservation of waves $q_t + c_y=0$ (c.f. \eqref{eq:11}) to verify that there is an equivalent modulation conservation law in $x$ having the form
\begin{equation}
    \label{eq:cons_law_rewrite}
    \left(\frac{f}{q} \right)_t+\left(\frac{fc}{q}-h \right)_x = 0.
\end{equation}
With \eqref{eq:cons_law_rewrite}, one can extend all of the above results in $y$ to results in $x$.

\subsection{Properties of the Y soliton shock solution}
\label{sec:unique}
\subsubsection{Construction and uniqueness of Y soliton shocks}
In this section we will discuss how to construct the Y soliton using modulation conservation laws and make some remarks regarding uniqueness. For three line solitons travelling together there are eight free variables: six soliton parameters and the velocity $(s_x,s_y)$. As the Y soliton is a three-parameter family of solutions to the KP equation \eqref{eq:kp}, at least five equations are required to completely determine the Y soliton. The speed relation \eqref{eq:rotate_spd} and two jump conditions for conservation of waves \eqref{eq:11} yield three independent equations. We obtain additional equations using the modified jump conditions in definition~\ref{def:mod_rh_cond} applied to  \eqref{eq:ave_cons} and the infinite family of conservation laws. The eight unknown parameters of the Y soliton solve this infinite system of algebraic equations. 

A discontinuous solution to the modulation conservation laws as described above is said to be \emph{admissible} if it corresponds to the zero dispersion limit of a travelling wave solution to the KP equation \eqref{eq:kp}. The Y soliton is the only KP travelling wave solution consisting of three solitons joined together, so any other discontinuous solution to the modulation conservation laws with three solitons is therefore inadmissible. It remains an open question as to how many modulation conservation laws are necessary to uniquely determine the Y soliton solution, i.e.  what is the minimal subset of the infinite set of algebraic equations that does not yield inadmissible solutions. 

The Hilbert basis theorem implies that a finite subset of equations exists that gives necessary and sufficient conditions to determine the unique solution to the entire infinite set of equations (see e.g. \cite{madhi_hybrid_2017}). The simplest, smallest candidate for this subset consists of \eqref{eq:rotate_spd}, two equations from jump conditions for conservation of waves \eqref{eq:11}, and modified jump conditions for conservation of momentum in $x$ \eqref{eq:ave_cons_1} and $y$ \eqref{eq:ave_cons_2}. Numerical calculations reveal that these five equations have a second (irrational) solution family that is inadmissible. Including the modified jump condition for conservation of energy \eqref{eq:ave_cons_3} does not eliminate this inadmissible family, but the next conservation law \eqref{eq:n_3_conserv} does eliminate it. Whether or not this set of six algebraic equations is sufficient to guarantee a unique solution corresponding to the admissible Y soliton is a question left for further study.

Using a similar approach, we can show that \emph{no two half solitons} can be connected by a single shock satisfying modified Rankine-Hugoniot jump conditions. First, no KP travelling wave solution consisting of two half solitons joined together exists, other than a soliton (trivially), so no corresponding  discontinuous solution can be admissible. Second, attempting to solve for the parameters of such a solution using conservation laws yields a combination of parameters with nonzero imaginary parts or negative amplitude. Consider a shock solution to the modulation system \eqref{eq:mod_syst} consisting of two half solitons
  \begin{equation}
  \label{eq:36}
  a(y,t) =
  \begin{cases}
    a_2 & y > s_y t \\ a_1 &  y< s_y t
  \end{cases}, \qquad q(y,t) = \begin{cases}
    q_2 & y > s_y t \\ q_1 &  y< s_y t
  \end{cases}, 
\end{equation}
where $a_1,a_2 > 0$. In order for a shock to exist in the modulation variables, ordinary jump conditions arising from conservation of waves \eqref{eq:phase_jump_cond_y}, as well as modified jump conditions for the conservation laws \eqref{eq:ave_cons_1} and \eqref{eq:ave_cons_2} should all equal $s_y$. Setting these three jump conditions equal to each other eliminates $s_y$ and gives two polynomial equations in $q$ and $\sqrt{a}$.  Given some lower soliton with parameters $(a_1,q_1)$, we obtain a nonlinear system with two equations and two unknowns. Solving for $(a_2,q_2)$ yields the following solutions for $a_2$
\begin{equation}
\label{eq:inexistence}
    a_2 \in \bigg\{\pm a_1, \left(-\frac{1}{2}\pm \frac{i}{2}\pm \sqrt{-1 \pm \frac{i}{2}} \right)a_1 \bigg\},
\end{equation}
where each $\pm$ should be understood to operate independently from the others, yielding a total of ten unique solutions for $a_2$. Given that $a_1>0$, each solution for $a_2$ in \eqref{eq:inexistence} is either negative or complex and is therefore not admissible. Consequently, it is impossible for initial conditions as given in \eqref{eq:36} to satisfy the modulation jump conditions unless $a_1=a_2$, $q_1=q_2$. Any shock solution to the modulation system \eqref{eq:mod_syst} must therefore consist of at least three half solitons joined together. 

\subsubsection{Y soliton as shock in Riemann invariants}
We find some additional properties of the Y soliton modulation shock by examining the $x$-independent modulation system \eqref{eq:mod_syst}, which can be written in diagonal form as \cite{ryskamp_2020}
    \begin{equation}
    \label{eq:RI_y}
        \frac{\partial}{\partial t }R_\pm + V_\pm \frac{\partial}{\partial y} R_{\pm} = 0, \qquad R_\pm = q \pm \sqrt{a}, \qquad V_{\pm}=2q \pm \frac{2}{3}\sqrt{a}=\frac{4}{3}R_\pm + \frac{2}{3}R_\mp \, ,
        \end{equation}
where $R_\pm$ are Riemann invariants and $V_\pm$ are the corresponding characteristic velocities. As pointed out in \cite{ryskamp2020_mf}, $R_\pm$ above are identical to the phase variables $k\pm$ for a soliton \eqref{e:inversemap}. As a result, for the Y soliton, one Riemann invariant is conserved for the jump from the Y soliton stem (leg 2 in figure~\ref{fig:soli_solns}$(b)$) to the other legs (leg 1 or 3). From leg 2 to leg 3, $R_+=k_3$ is conserved. Examining characteristic velocities \eqref{eq:RI_y} shows that this is a \emph{compressive} shock in the slow characteristic family (i.e. $V_-(y>s_y t)<V_-(y<s_y t)$). From leg 2 to leg 1 in figure~\ref{fig:soli_solns}$(b)$, $R_-=k_1$ is conserved, which in this case is an \emph{expansive} shock in the fast characteristic family (i.e. $V_+(y>s_y t)>V_+(y<s_y t)$). For the jump from leg 1 to leg 3, neither Riemann invariant is conserved. This jump in $x$ passes through $q=0$, where the $y$-independent modulation system is nonstrictly hyperbolic, and in fact the $y$-independent soliton ansatz cannot be defined for $q=0$ (cf. \cite{ryskamp2020_mf}).

\par The conservation of Riemann invariants gives us one more way to calculate the velocity of the Y soliton $s_y$. For the jump where $R_-=k_1$ is conserved, from leg 1 to leg 2, we can use ordinary jump conditions for the hyperbolic equation for the Riemann invariant in $y$ \eqref{eq:RI_y} to correctly calculate $s_y$, ignoring leg 3. A similar calculation holds for conservation of $R_+$ from leg 2 to leg 3.

\section{Application: V-shape initial conditions}
\label{sec:v_shape}
The previous sections showed that the Y soliton can be understood as a shock solution to an infinite family of modulation conservation laws. In this section we apply these insights to V-shape initial conditions, i.e. the Mach reflection problem, which in modulation variables is written as
\begin{equation}
    \label{eq:ArrowheadIC}
    a(y,0)=a_0, \qquad
    q(y,0)=
    \begin{cases} 
        -q_0 & y>0 \\
        q_0 & y<0 
    \end{cases}, \quad
\end{equation} where $q_0>0$. The initial conditions \eqref{eq:ArrowheadIC} model the symmetrical oblique interaction of two line solitons, which is identical to a single soliton encountering an inward oblique corner through a method of images argument \cite{yeh_2010,kodama_yeh_2016} (see schematic in figure~\ref{fig:mach_ref_schematic}). This problem is well-studied \cite{kodama_2009,yeh_2010,kodama_yeh_2016,kodama_book,chakravarty_numerics_2017}, but the following modulation theory approach offers new insight on the dynamics of the system evolution. Given a modulation solution $a(y,t)$ and $q(y,t)$ to \eqref{eq:ArrowheadIC}, the solution of the full KP equation \eqref{eq:kp} can be reconstructed via \begin{equation}
  \label{eq:39}
    u(x,y,t) = a(y,t) \,\mathrm{sech}^2\left (
      \sqrt{\frac{a(y,t)}{12}} \theta \right ), \qquad
    \theta = x + \int_0^y q(y',t)\,\mathrm{d}y' - \int_0^t
    c(0,t')\,\mathrm{d} t' ,
\end{equation}
where the soliton velocity satisfies $c(y,t) = a(y,t)/3 + q(y,t)^2$ (cf.~\eqref{eq:soli}, \eqref{eq:x_indep_soli}). Numerical evolutions of the full KP equation \eqref{eq:kp} with $\epsilon=1$ and initial conditions \eqref{eq:ArrowheadIC} projected onto \eqref{eq:39} (with smoothing of $q(y,0)$ to prevent Gibbs phenomenon) are shown in figure~\ref{fig:v_shape_stem} for $\sqrt{a_0}>q_0$ and in figure~\ref{fig:reg_ref} for $\sqrt{a_0}<q_0$. Note that the scale-free initial condition \eqref{eq:ArrowheadIC} and modulation equations \eqref{eq:full_mod_syst} imply that the $\epsilon \to 0$ limit for $t = \mathcal{O}(1)$ is equivalent to the $t \to \infty$ limit with $\epsilon = \mathcal{O}(1)$.
\begin{figure}
    \centering
    \includegraphics{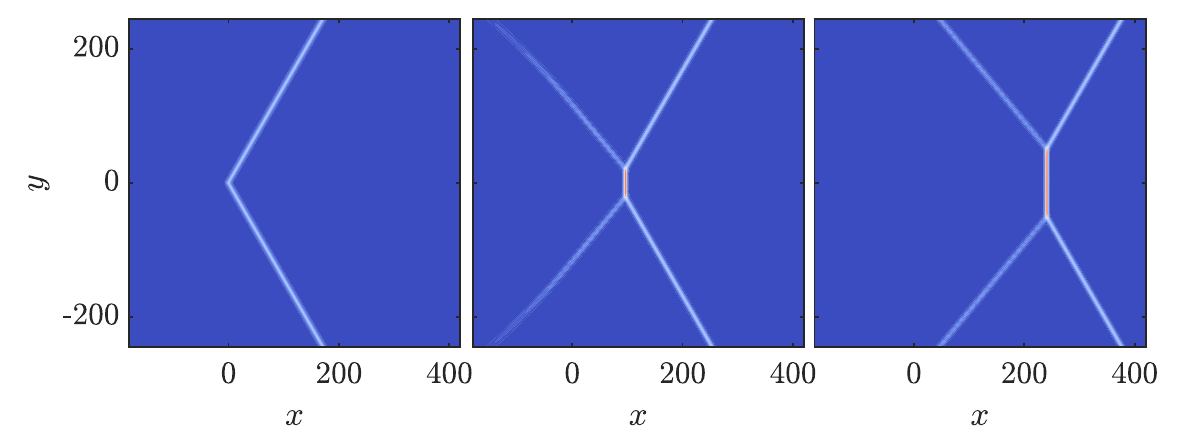}
    \caption{Numerical evolution of V-shape initial conditions  $\eqref{eq:ArrowheadIC}$ projected onto \eqref{eq:39} for the KP equation \eqref{eq:kp} with $\epsilon=1$ for $a_0=1$ and $q_0=.7$ on $t \in \{0,100,250\}$.}
    \label{fig:v_shape_stem}
\end{figure}

\subsection{Mach reflection}
By directly examining the characteristic velocities \eqref{eq:RI_y}, it is evident that the solution to \eqref{eq:ArrowheadIC} must give rise to shocks in the modulation variables. For $y>0$ in \eqref{eq:ArrowheadIC} the characteristic velocities are $V_{\pm}(y>0)=-2q_0 \pm 2 \sqrt{a_0}/3$, while for $y<0$, $V_\pm(y<0)=2q_0\pm 2 \sqrt{a_0}/3$. Since $q_0>0$, $V_\pm(y<0)>V_\pm(y>0)$, so both characteristic velocities are compressive and should give rise to shocks in the modulation equations. As Y solitons are shocks in modulation variables, we will look for solutions to \eqref{eq:ArrowheadIC} consisting of two Y solitons.

Initial data \eqref{eq:ArrowheadIC} for the $x$-independent modulation equations \eqref{eq:11} and \eqref{eq:conserv_laws_mod} constitute a Riemann problem for a hyperbolic system of two equations. Traditionally, such problems are solved in terms of two waves---one slow and the other fast---that are separated by a constant region. Each wave is either a shock or rarefaction wave. Here, the multivalued structure of the Y shock and the modified Rankine-Hugoniot conditions lead us to solve this Riemann problem with four waves separated by constant regions. The four waves are visible in the KP simulation of figure~\ref{fig:v_shape_stem}. We observe the two reflected Y shocks separated by a constant stem $(a_{\rm i},q_{\rm i})$ and two rarefaction waves that lead the growth of two new partial soliton legs from $a=0$ to $(a,q)=(a_{\rm n},q_{\rm n})$ (upper left leg) and $(a,q)=(a_{\rm n},-q_{\rm n})$ (lower left leg). The new legs arise as a consequence of the fact that modulation shocks must be composed of at least three soliton legs to satisfy the jump conditions (see section~\ref{sec:unique}). 
\begin{figure}
    \centering
    \includegraphics[scale=.25]{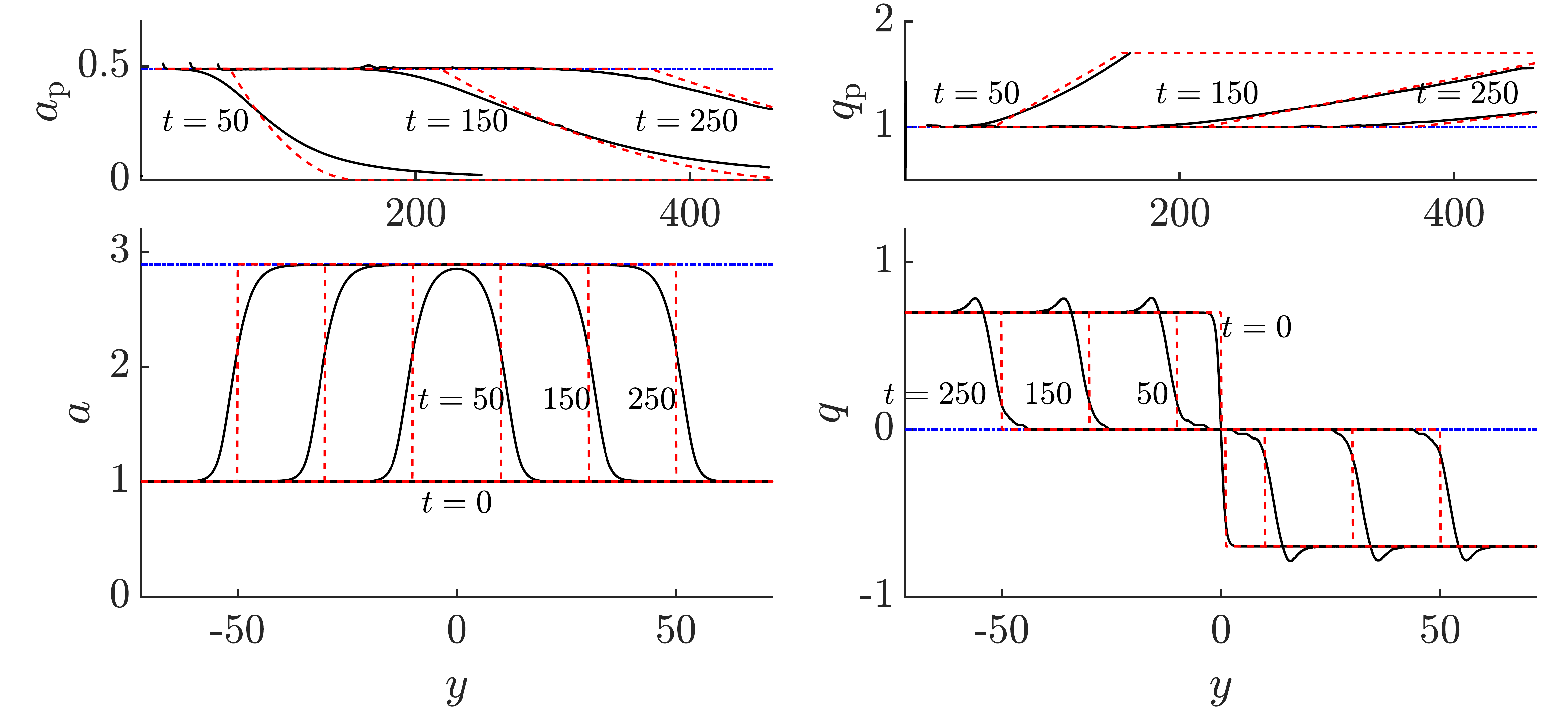}
    \caption{Comparison between analytical solution (dashed) with numerical solution (solid) for v-shape initial conditions from figure~\ref{fig:v_shape_stem} for $a$ (left) and $q$ (right) on $t \in \{0,50,150,250\}$. The top panel shows the evolution of the upper left leg as modelled by \eqref{eq:38}, and the bottom panel shows the evolution of the original legs and the shock stem as described by \eqref{eq:v_shape_soln}. In all cases the evolution approaches the long-time prediction (dashed-dotted). The top plot includes a phase shift of $y_0 = 12$ in order to account for higher-order effects. }
    \label{fig:exact_vs_num_mach}
\end{figure}

We can determine the parameters of the stem $(a_{\rm i},q_{\rm i})$ by looking for a state that shares one phase variable (Riemann invariant) with the top and the bottom, solving the equations
\[ -q_0-\sqrt{a_0}=q_{\rm i}-\sqrt{a_{\rm i}},\qquad q_0+\sqrt{a_0}=q_{\rm i}+\sqrt{a_{\rm i}}.\]
We obtain that $q_{\rm i}=0$ and $\sqrt{a_{\rm i}}=\sqrt{a_0}+q_0$. From the requirement that the physical variables of a Y soliton conserve phase parameters (see section~\ref{sec:unique}), one can determine that the new upper left leg has parameters $(a_{\rm n},q_{\rm n})=(q_0^2,\sqrt{a_0})$, with symmetrical results for the bottom new leg. Using the jump conditions calculated throughout this paper, we also determine that the velocity of the top discontinuity is $s_y=\frac{2}{3}(\sqrt{a_0}-q_0)$. Consequently, leaving aside the growth of the extra leg on the top and bottom, the evolution of the initial half solitons and the new stem becomes
\begin{subequations}
    \label{eq:v_shape_soln}
    \begin{equation}
    \label{eq:47}
    a(y,t)=\begin{cases}
        a_0 & |y| > \frac{2}{3}(\sqrt{a_0}-q_0)t \\
        (q_0+\sqrt{a_0})^2 & |y| < \frac{2}{3}(\sqrt{a_0}-q_0)t
        \end{cases} ,
    \end{equation}
    \begin{equation}
    \label{eq:48}
    q(y,t)=
    \begin{cases} 
        -\mathrm{sgn}(y) q_0 & |y| > \frac{2}{3}(\sqrt{a_0}-q_0)t \\
        0 & |y| < \frac{2}{3}(\sqrt{a_0}-q_0)t
    \end{cases}.
    \end{equation}\end{subequations}
Comparing the modulation theory analytical solution \eqref{eq:v_shape_soln} with parameter values extracted from numerical simulation of the full KP equation \eqref{eq:kp} with $\epsilon=1$, shown in figure~\ref{fig:v_shape_stem}, we see that \eqref{eq:v_shape_soln} shows excellent agreement with large $t$ numerical results in the bottom panels of figure~\ref{fig:exact_vs_num_mach}.
\par Modulation theory can also analytically describe the growth of the new leg. We focus on the upper left leg, with symmetrical results for the bottom left leg. Since this leg must appear immediately to satisfy the jump conditions, we can model its appearance as a partial soliton, a problem previously solved \cite{neu_singular_2015,ryskamp_2020,ryskamp2020_mf}. The partial soliton problem is a Riemann problem in the modulation variables
\begin{equation}
    \label{eq:partial_soli}
       a_{\rm p}(y,0)=\begin{cases} a_{\rm n} & y<0 \\ 0 & y>0 \end{cases}, \qquad
    q_{\rm p}(y,0)=\begin{cases} q_{\rm n} & y<0 \\ q_* & y>0 \end{cases},
\end{equation}
where $q_*$ can be determined by conservation of a Riemann invariant. The solution to \eqref{eq:partial_soli} is a simple wave with $R_+$ constant, and we can write down an explicit solution for all values above the expanding stem where $a(y,t)$ and $q(y,t)$ are multivalued, i.e. for $y>\frac{2}{3}(\sqrt{a_0}-q_0)t$,
\begin{subequations}
\label{eq:38}
\begin{align}
    \sqrt{a_{\rm p}(y,t)} &=
    \begin{cases}
      0 & \qquad \qquad \quad \, 2q_*t <  y \\ \, \frac{3}{8}\left(2q_*-\frac{y}{t}\right) &
      \left(2q_{\rm n}-\frac{2}{3}\sqrt{a_{\rm n}}\right)t < y < 2q_*t  \\ a_{\rm n} & \quad \frac{2}{3}(\sqrt{a_0}-q_0)t < y < \left(2q_{\rm n}-\frac{2}{3}\sqrt{a_{\rm n}}\right)t
    \end{cases}, \\
    q_{\rm p}(y,t) &= q_{\rm n}+\sqrt{a_{\rm n}} - \sqrt{a_{\rm p}(y,t)}.
\end{align}\end{subequations}
Here $q_*=R_+=q_{\rm n}+\sqrt{a_{\rm n}}$ and represents the limit of $q(y,t)$ for the $a=0$ edge of the simple wave. Note that the new leg is well-defined, since the bottom (soliton) edge of the simple wave travels with velocity $2q_{\rm n}-\frac{2}{3}\sqrt{a_{\rm n}}=2\sqrt{a_0}-\frac{2}{3}q_0>\frac{2}{3}(\sqrt{a_0}-q_0)$, so the upper left leg always grows faster than the discontinuity with the stem. The solution \eqref{eq:38} is shown to accurately capture the large $t$ numerics in the top panels of figure~\ref{fig:exact_vs_num_mach}. To our knowledge, the evolution \eqref{eq:38} has not been obtained in previous work for V-shape initial conditions \eqref{eq:ArrowheadIC}.

The above \emph{Mach reflection} behaviour occurs for $\sqrt{a_0} \ge q_0$. The maximum amplitude of the intermediate (stem) region is obtained when $q_0 = \sqrt{a_0}$, which gives $a_i=(2\sqrt{a_0})^2=4a_0$, precisely the maximum four-fold amplification first predicted by Miles \cite{Miles1977oblique}. When $q_0>\sqrt{a_0}$, the solution \eqref{eq:v_shape_soln} breaks down, since the stem no longer grows. We consider this case next. 
\subsection{Regular reflection}
\begin{figure}
    \centering
    \includegraphics[scale=1]{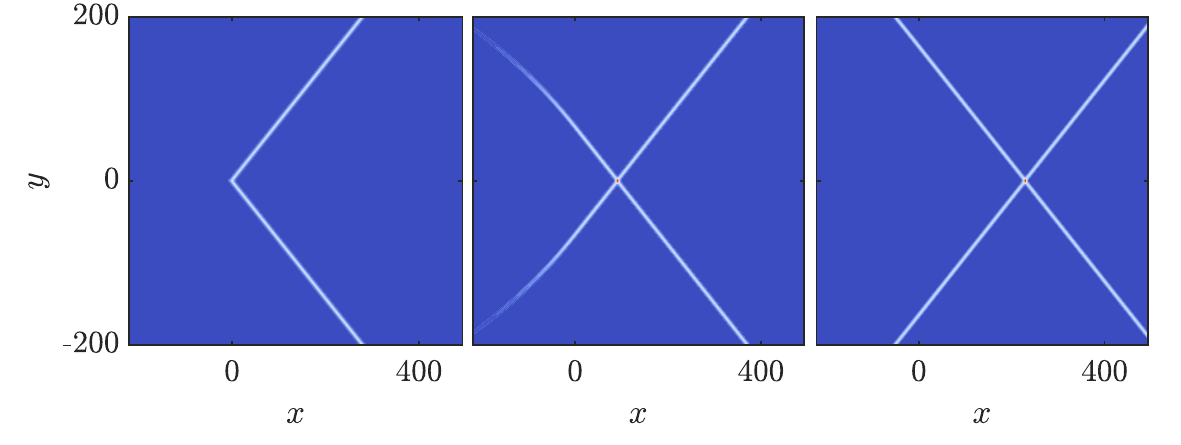}
    \caption{Numerical evolution of the KP equation \eqref{eq:kp} with $\epsilon=1$ for V-shape initial conditions $\eqref{eq:ArrowheadIC}$ projected onto \eqref{eq:39} on $t \in \{0,40,100\}$ for $a_0=1$ and $q_0=1.4$. The resonance condition is not met, so Mach reflection does not occur.}
    \label{fig:reg_ref}
\end{figure}
When $q_0>\sqrt{a_0}$ for \eqref{eq:ArrowheadIC}, we have \emph{regular reflection}. It is well known that in this regime an X-shaped solution arises where two solitons intersect with a small phase shift \cite{kodama_2009}, which is higher-order \cite{Miles1977oblique} and therefore not captured by modulation theory. A simulation of regular reflection is shown in figure~\ref{fig:reg_ref}. The growth of the new left legs of the X can be modelled identically as before in \eqref{eq:38}, with the only difference that $a_{\rm n}=a_0$ and $q_{\rm n}=-q_0$, and the lower bound of the new soliton is $y=0$. Consequently, $q_*=q_0+\sqrt{a_0}$. The partial soliton solution \eqref{eq:38} with the new parameters is favourably compared with numerical simulation in figure~\ref{fig:reg_ref_acc}. 
\par The resulting X-shaped soliton, often called the O-type (meaning ``ordinary'' \cite{kodama_book}) soliton, is a travelling wave exact solution of the KP equation \eqref{eq:kp} with velocity $(\overline{s_x},\overline{s_y})$. By the argument in section~\ref{sec:shock_soln}, this structure as a whole should satisfy the jump conditions as given in section~\ref{sec:shock_soln}. For $y>\overline{s_y}t$, we have two solitons with parameters $(a_0,\pm q_0)$ while for $y<\overline{s_y}t$, we also have two solitons with parameters $(a_0,\mp q_0)$. As a result, all $(n,m)=(2,2)$ modified Rankine-Hugoniot jump conditions are trivially satisfied with no information gained about $s_y$. However, we can still utilize jump conditions for $q$ from \eqref{eq:11} to determine the velocity $\overline{s_y}$. Taking any two legs with opposite signs for $q$, we obtain
\[
    \overline{s_y} = \frac{\llbracket a/3+q^2 \rrbracket}{\llbracket q \rrbracket} = \frac{a_0/3 + q_0^2-(a_0/3+q_0^2)}{q_0-(-q_0)}=0, 
\]
as expected from symmetry.

\begin{figure}
    \centering
    \includegraphics[scale=.25]{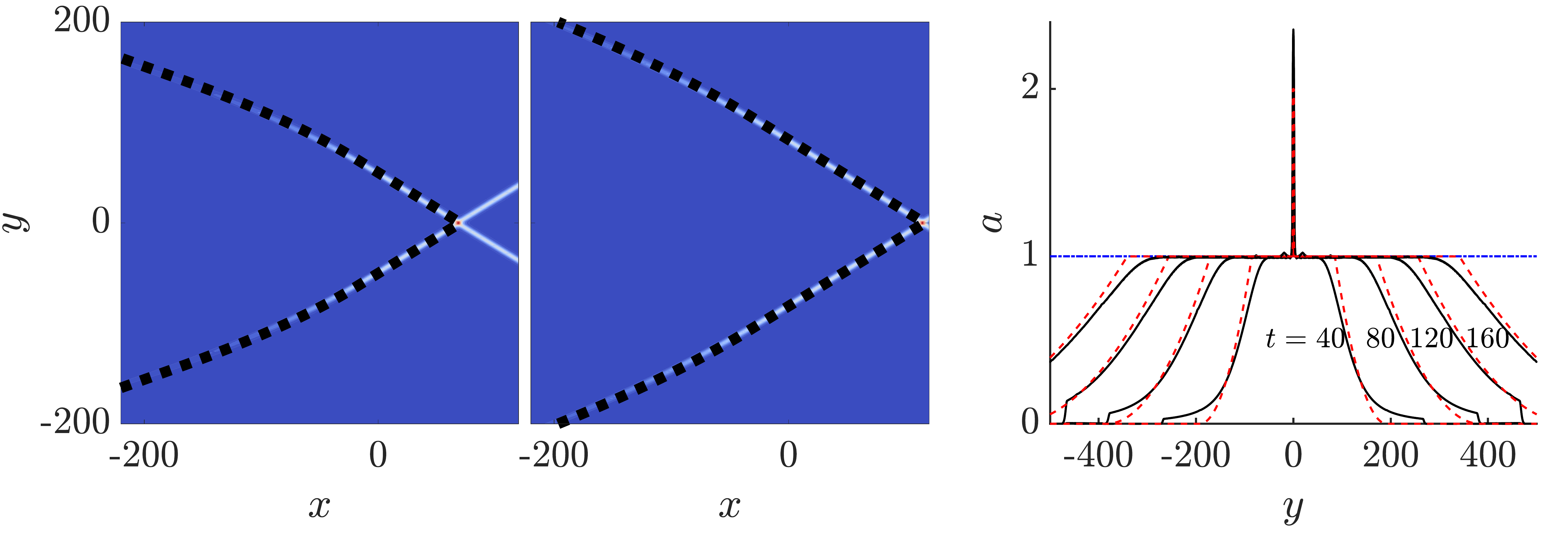}
    \caption{Accuracy of the modulation prediction for the simulation from figure~\ref{fig:reg_ref}. The left two panels show contour plots for $t = 30$ and $t=50$, with the predictions for $\theta$ from \eqref{eq:38} and \eqref{eq:39} overlayed (dashed). The right panel shows the accuracy of the amplitude prediction (dashed) compared to numerical simulation (solid) for $t \in \{40,80,120,160\}$. The long-time asymptotic prediction (dash-dotted) is also shown. The modulated prediction is shown to be accurate, although the peak amplitude is actually $\approx 2.3$, whereas superposition of modulation theory solutions predicts $a(0,t) = 2$. }
    \label{fig:reg_ref_acc}
\end{figure}

\section{Discussion and conclusion}
\label{sec:concl}
\par The main finding of this paper is that the KP resonant Y soliton can be modelled as a shock solution to the soliton modulation equations. Remarkably, the velocity of the Y soliton satisfies modified Rankine-Hugoniot jump conditions for an infinite family of conservation laws. The justification for our result is the fact that the modulation equations \eqref{eq:mod_syst} represent the soliton parameter evolution in the zero dispersion limit of the full KP equation \eqref{eq:kp}. In this limit, the rapid transition from the single, larger stem of the resonant soliton to the two smaller legs limits to a single point of discontinuity. This zero dispersion limit approach for defining admissible discontinuous weak solutions is analogous to the vanishing viscosity approach of admissible classical shocks in conservation law theory \cite{lefloch_shock_book}. Our work is the second known example of Whitham shocks \cite{sprenger_whitham_shocks}, discontinuities in the modulation variables that correspond to travelling wave solutions of the full equation. 

Since the Y soliton parameters are multivalued in the modulation equations, and the fully two-dimensional modulation system reduces to an equivalent one-dimensional system, we must introduce modified Rankine-Hugoniot conditions. Namely, the fluxes and densities of the two smaller legs are added when calculating the shock velocities. This is justified by the fact that the two smaller legs summed together satisfy integral conservation laws when one integrates over $x$. 

 Through Whitham modulation theory, we were able to analytically describe the evolution of V-shape initial conditions \eqref{eq:ArrowheadIC}. Unlike prior work, our approach also determines the dynamic growth of the reflected legs as partial solitons. These predictions were quantitatively verified to high accuracy in the large $t$ regime by numerical simulation for the KP equation \eqref{eq:kp}. In addition, we identified the transition between Mach reflection and regular reflection for a soliton incident upon an inward, compressive corner, with Mach reflection occurring when $|q_0|<\sqrt{a_0}$.
 
 Mach reflection is analogous to the phenomenon of Mach expansion, identified in \cite{ryskamp_2020}, where a soliton incident on an \emph{outward} corner creates a new, straight \emph{lower-amplitude} stem whenever $|q_0|<\sqrt{a_0}$. Within the Mach reflection and expansion regimes, a new line soliton is formed perpendicular to the wall. Outside of these regimes, the interaction of a soliton with an outward or inward corner is very different, with the initial partial solitons evolving independently from each other.

 The results in this paper are (to our knowledge) the first application of modulation theory to a fully two-dimensional evolution, thereby opening new doors for future work studying soliton resonance in other two-dimensional equations. Modulation theory does not depend on integrability or, relatedly, the existence of a large family of multi-soliton solutions. Consequently, this work can be generalized to any system with a known line soliton solution. For example, the two-dimensional Benjamin-Ono (2DBO) equation, which models internal water waves \cite{ablowitz_segur_1980}, has numerically-observed resonance-like behaviour \cite{numerics_2001} and line soliton solutions for which modulation equations have been calculated \cite{Ablowitz_whitham_2dbo}. Few analytical tools and no known exact multi-soliton solutions exist for the 2DBO equation, but the generalization of our approach seems like a promising avenue for research in this area.

\section*{Acknowledgments}
Authors thank Jonathan Hauenstein for helpful discussions concerning polynomial systems. The work of MAH and SR was supported by NSF grant DMS-1816934. The work of GB was
supported by NSF grant DMS-2009487.

\appendix
\section{Reducibility of multidimensional modulations}
\label{sec:appendix}
In this appendix we prove Theorem~\ref{def:thm}. 
Namely, we will show that when $a$ and $q$ describe a modulated KP soliton \eqref{eq:soli}, the system \eqref{eq:full_mod_syst} can be reduced to a (1+1)-dimensional system, assuming differentiable modulations.
\begin{proof} 
The full modulation system \eqref{eq:full_mod_syst} is equivalent to \eqref{eq:mod_syst} under the assumption that $a_x \equiv q_x \equiv 0$, which followed from the ansatz \eqref{eq:x_indep_soli}. This assumption can be relaxed by using a more general line soliton ansatz as
\begin{equation}
    \label{eq:gen_mod_soli}
u_0(\theta,x,y,t; \epsilon) =  a(x,y,t)\sech^2\left(\frac{1}{\psi}\sqrt{\frac{a(x,y,t)}{12}}\frac{\theta}{\epsilon}\right), 
\end{equation}
where $\psi = \psi(a,q)$ is nonzero, real-valued, and differentiable. If the soliton is unmodulated (i.e. $a(x,y,t)\equiv a$ and $q(x,y,t)\equiv q$), \eqref{eq:gen_mod_soli} must be consistent with \eqref{eq:soli}. Since the argument of $\sech^2$ in \eqref{eq:gen_mod_soli} has been divided by $\psi$, it must also be multiplied by $\psi$. We incorporate $\psi$ into the derivatives of the fast phase $\theta/\epsilon$ so that
\begin{equation}
    \label{eq:gen_theta_derivs}
    \theta_x = \psi(a,q), \qquad \theta_y = q \psi(a,q), \qquad \theta_t = -\psi(a,q)\left(\frac{a}{3}+q^2\right),
\end{equation}
where $a=a(x,y,t)$ and $q=q(x,y,t)$ are the modulation parameters. For the ansatz in this work \eqref{eq:x_indep_soli}, we selected $\psi = 1$, while for a previous work focusing on $y$-independent modulations we selected $\psi = 1/q$ \cite{ryskamp2020_mf}. The  compatibility conditions $\theta_{xy}=\theta_{yx}$ and $\theta_{xt}=\theta_{tx}$ for \eqref{eq:gen_theta_derivs} yield, respectively
\begin{subequations}
\begin{align}
    \label{eq:155}
    \psi_a a_y + \psi_q q_y &= q \psi_a a_x + (q \psi_q + \psi) q_x, \\
    \label{eq:156}
    \psi_a a_t + \psi_q q_t &= -\left(c \psi_a + \frac{\psi}{3} \right)a_x - \left(c \psi_q + 2q\psi \right)q_x. 
\end{align}\end{subequations}
We require consistency with the full modulation system \eqref{eq:full_mod_syst}, so we insert the equations for $a_t$ and $q_t$ from \eqref{eq:full_mod_syst} into \eqref{eq:156} to obtain
an equation in which $(a,q)$ only depend on the spatial variables $(x,y)$. Using simplifying calculations together with \eqref{eq:155}, we obtain the system
\begin{equation}
\label{eq:6p8}
    \begin{bmatrix} 
    \psi_a & \psi_q \\ \psi_q & 4a \psi_a
    \end{bmatrix} 
    \begin{bmatrix} a_y \\ q_y 
    \end{bmatrix} =    
    q\begin{bmatrix} 
     \psi_a &  \psi_q \\  \psi_q & 4 a  \psi_a
    \end{bmatrix} 
    \begin{bmatrix} a_x \\ q_x 
    \end{bmatrix} +     \begin{bmatrix} 
     0 &  \psi\\  \psi & 0
    \end{bmatrix} \begin{bmatrix} a_x \\ q_x 
    \end{bmatrix} \, .
\end{equation}
We can always solve for either $(a_x,q_x)$ or $(a_y,q_y)$ in equation \eqref{eq:6p8}, since for any $\psi$ either the left hand side or the right hand side matrix in \eqref{eq:6p8} is invertible. We show this by assuming both determinants are zero and generating a contradiction. Let us rewrite \eqref{eq:6p8} as $A \bm{v}_y = B \bm{v}_x$, where $\bm{v}=[a\; q]^T$ and $A,B$ are the appropriate matrices from \eqref{eq:6p8}. We assume that 
\begin{subequations}
\label{eq:assump}
\begin{align}
    \label{eq:assump1}
    \det{A} &= 4a \psi_a^2- \psi_q^2 = 0, \\
     \label{eq:assump2}
     \det{B} &= 4a \psi_a^2- \psi_q^2 -\psi^2-2q\psi_q \psi = -\psi^2-2q\psi_q \psi = 0.
\end{align}\end{subequations}
The condition \eqref{eq:assump2} for nonzero $\psi$ implies that $\psi+2q\psi_q=0$. This ordinary differential equation is solved as $\psi=R(a)/\sqrt{q}$ for any $R(a)$. But inserting this into \eqref{eq:assump1}, we obtain an equation for $R$
that is satisfied when 
$ 4 q \sqrt{a} R'(a) = \pm R(a)$.
 This contradicts the fact that $R=R(a)$ is only a function of $a$, so we cannot find a solution for $\psi$ that satisfies \eqref{eq:assump} and sets both determinants to zero. Therefore, for any $\psi$ either $A$ or $B$ must have a nonzero determinant and therefore be invertible. 
 \end{proof}
 \par Consequently, in \eqref{eq:6p8} we can solve for either $(a_x,q_x)$ or $(a_y,q_y)$. We can then locally substitute whichever variables were solved for into the full modulation system \eqref{eq:full_mod_syst}, reducing \eqref{eq:full_mod_syst} to a (1+1)-dimensional system.

\clearpage
\bibliography{refs} 

\end{document}